\documentclass[11pt,a4paper,final]{article}
\usepackage[latin1]{inputenc}
\usepackage[english]{babel}
\usepackage{amsmath}
\usepackage{amsthm}
\usepackage{mdwlist}
\usepackage{enumerate}
\usepackage{amsfonts}
\usepackage{amssymb}
\usepackage{tikz}
\usetikzlibrary{patterns}

\usepackage[left=3.0cm, right=3.0cm, bottom=4cm]{geometry}

\newtheorem{theorem}{Theorem}
\newtheorem{lemma}[theorem]{Lemma}
\newtheorem{proposition}[theorem]{Proposition}

\theoremstyle{remark}
\newtheorem{remark}[theorem]{Remark}

\newcommand{\fz}{f(z,\lambda) }      
\newcommand{\im}{\mathrm{i}}
\newcommand{\e}{\mathrm{e}}
\renewcommand{\Re}{\operatorname{Re}}
\renewcommand{\Im}{\operatorname{Im}}

\title{The real spectrum of the imaginary cubic oscillator: An expository proof}
\author{Ilario Giordanelli, Gian Michele Graf \\Theoretische Physik, ETH Zurich, 8093 Zurich, Switzerland}
\begin{document}
\maketitle
\begin{abstract}
We give a partially alternate proof of the reality of the spectrum of the imaginary cubic oscillator in quantum mechanics.
\end{abstract}
\maketitle

\section{Introduction}
The imaginary cubic oscillator is defined by the Hamiltonian 
\begin{equation}
\label{hamiltonian}
H=p^2+\im x^3
\end{equation}
with $p=-\im d/dx$, which is manifestly not self--adjoint. Its spectrum is nevertheless real, as conjectured by Bessis and Zinn-Justin in 1992, numerically confirmed by Bender and Boettcher~\cite{beboe}, proven by Dorey et al.~\cite{dorey}, as well as by Shin~\cite{shin}. Recent related works are \cite{grema, henry, sk, zjj}, among others. 

The Hamiltonian has become a paradigmatic model of $PT$-symmetric quantum mechanics (see \cite{bender} for an introduction and \cite{most} for a critique). The joint operation $PT$ of parity $P$ and time--reversal $T$ is anti-unitary, acts as $(PT\psi)(x)=\overline{\psi(-x)}$, and leaves $H$ invariant. The spectrum of $H$ is invariant under complex--conjugation, that conclusion being easiest when the spectrum is discrete, as it is here. In fact, $PT$ maps the eigenspace of an eigenvalue to that of its complex--conjugate. If in addition all eigenstates of $H$ are also eigenstates of $PT$, then the spectrum is real, and the $PT$-symmetry is called unbroken. As mentioned, this is known to be the case for Eq.~(\ref{hamiltonian}).

The modest goal of this paper is to provide an expository proof of that result. To keep it reasonably short, we refrain from giving any extensions. The proof is self--contained except for using tools of general character. Its novelty is limited to technical aspects, with differences and similarities to previous proofs discussed in the next section.  

The full result is as follows. 

\begin{theorem}
\label{main_thm} Let the operator $H$ on $L^2(\mathbb{R})$ be given as by (\ref{hamiltonian}) with domain $\mathcal{D}(H)=\mathcal{D}(p^2)\cap\mathcal{D}(x^3) $. 
Then 
\begin{enumerate}[(a)]
\item $H$ is closed,
\item its spectrum is discrete 
\item with algebraically simple eigenvalues that 
\item are real and positive,
\item and infinite in number. Moreover, the eigenvalues are given
 asymptotically as $\lambda_n=\lambda^0_n +o(1)$, ($n\to\infty$),
 where
\begin{gather}\label{l0n}
\lambda^0_n=\Bigl(\frac{(2n+1)\pi}{\sqrt{3}K}\Bigr)^{6/5},\\
\label{K1}
K=2\int_0^1 \sqrt{1-t^3}dt.
\end{gather}
\end{enumerate}
\end{theorem}
The main point is the unbroken $PT$-symmetry, meaning the reality of
the eigenvalues. Item (c) is due to \cite{trinh1, trinh2}, in that simplicity holds true even in the algebraic sense. Item (d) describes only the large eigenvalues; in particular it is not claimed that $\lambda_n$ is the $n$-th eigenvalue in increasing order when counting from $n=0$. However in \cite{bender} that claim is clearly supported numerically and, moreover, it is observed that the asymptotics is semi-classical, with Eq.~(\ref{l0n}) being the WKB-approximation for the eigenvalues. Among the three complex turning points $x$ at energy $\lambda>0$, i.e., $\im x^3=\lambda$, the two relevant ones are $x_+=\e^{-\im\pi/6}\lambda^{1/3}$ and $x_-=\e^{-5\pi\im/6}\lambda^{1/3}$, whereas $x_0=\im\lambda^{1/3}$ is not, because the solution $\psi_\pm(x)$ decaying at $x\to\pm\infty$ does so in a sector also containing $x_\pm$, but not $x_0$ (as may be inferred from Eq.~(\ref{Psi_pm}) below). The Sommerfeld condition is 
\begin{equation}
\label{somm}
2\int_{x_-}^{x_+}pdx=2\pi\bigl(n+\frac{1}{2}\bigr)
\end{equation}
with $p=\sqrt{\lambda-\im x^3}$. The path joining $x_-$ to $x_+$ can be taken as the polygonal chain $x_-, 0, x_+$. The l.h.s. then equals $(\e^{-\im\pi/6}-e^{-5\pi\im/6})K\lambda^{5/6}=\sqrt{3}K\lambda^{5/6}$ and the solutions of (\ref{somm}) are given by (\ref{l0n}).
\section{Main steps of the proof}\label{ms}
The structure of the proof owes a great deal to work \cite{shin} of
Shin and indirectly to \cite{dorey}, in particular in regard to the asymptotic properties of the solution of the eigenvalue equation viewed as an ordinary differential equation. Parts (a, b) can be traced back to \cite{cali}. The present approach differs however from \cite{shin} on two counts.

First, these properties are derived from a general result of Olver \cite{olver} on the Liouville-Green approximation, which in our opinion makes for a more straightforward route than the proof \cite{sibuya} cited in \cite{shin}. (The same observation was very recently made in \cite{henry}.) Second, these properties are used somewhat differently in order to conclude that $PT$-symmetry is unbroken, and in fact by way of the Phragm\'en-Lindel\"of principle. 

We focus on property (d) of Theorem~\ref{main_thm}, this being the main result. The proofs of parts (a-c) and (e) are deferred to Section 3; likewise those of the statements made below, if not immediately given. In view of (a), any solution of the eigenvalue equation $H\psi=\lambda \psi$ is a classical solution of the ordinary differential equation (ODE)
\begin{equation}
\label{DE}
-\psi''(x)+\im x^3\psi(x)=\lambda \psi(x)
\end{equation}
with $\psi \in L^2(\mathbb{R})$. We disregard that normalizability condition for the time being. The (so far) heuristic replacement 
\begin{equation}
\label{f_psi_relation}
\psi(x)=f(\im x)
\end{equation}
 turns the ODE into 
 \begin{equation}
 \label{DE_turned}
 -f''(z)+(z^3+\lambda)f(z)=0.
 \end{equation}
The following proposition (cf.~\cite{shin}, Prop.~5) states that a particular solution can be characterized by its asymptotics as $z$ tends to $\infty$ within the sector $|\arg{z}|<3\pi/5$.

\begin{proposition}
\label{Properties_of_sol_f}
Eq.~(\ref{DE_turned}) admits a solution $f(z,\lambda)$ with the following properties: 
\begin{enumerate}[i)]
\item $f(z,\lambda)$ is an entire function of $z$, $\lambda$.
\item $f$ has the asymptotics 
\begin{equation}
\label{asympt}
f(z,\lambda) =z^{-3/4}\e^{-(2/5)z^{5/2}}(1+O(z^{-1/2}))
\end{equation} 
as $z\to \infty$ within $|\arg{z}|\leq 3\pi/5-\varepsilon$, $(\varepsilon>0)$ uniformly on compact sets in $\lambda$.
\item Property (ii) uniquely determines the solution $f$. In particular, $$\overline{f(z,\lambda)}=f(\bar{z},\bar{\lambda}).$$
\item 
\begin{equation}
\label{asymp_f(0,lambda)}
f(0,\lambda)=\lambda^{-1/4}\e^{K\lambda^{5/6}}(1+O(\lambda^{-5/6}))
\end{equation}
as $\lambda \to \infty$ in the sector $|\arg{\lambda}|\leq
\pi-\delta$, $(\delta >0)$, with $K$ as in (\ref{K1}).
\item For any $\tilde{K}>K $ and some $C>0$, 
\begin{equation}
\label{upper_bondf}
|f(0,\lambda)|\leq C \e^{\tilde{K}|\lambda|^{5/6}},\qquad(\lambda\in\mathbb{C}).
\end{equation}
\end{enumerate}
\end{proposition}

Here and elsewhere roots are defined by the standard branch with cut $(-\infty, 0]\subset \mathbb{C}$.

We observe that $\fz$ is decaying in $z$ in the Stokes sector 
$$S_0=\lbrace z\in \mathbb{C} \mid |\arg{z}|<\pi/5\rbrace $$
and diverging in $S_\pm=\e^{\pm 2\pi \im/5 }S_0$. Further solutions of (\ref{DE_turned}) are obtained by substitution in the form of $f(\omega^{-1}z,\omega^{-3}\lambda)$ provided $\omega^5=1$, and hence $\omega^{-5}=1$. We pick $\omega=\e^{2\pi \im/5}$ and set 
\begin{equation}
\label{f_pm}
f_\pm(z,\lambda):=f(\omega^{\mp 1} z,\omega^{\mp 3}\lambda).
\end{equation} 
Thus $f_\pm$ is decaying in $S_\pm$ and diverging in $S_0$. Moreover, 
\begin{equation}
\label{barf+_f_-}
\overline{f_+(z,\lambda)}=f_-(\bar{z},\bar{\lambda}).
\end{equation} 
\begin{lemma}
\label{stokes_coeff} There is $C(\lambda)\in\mathbb{C}$ such that
\begin{equation}
\label{sc}
f_-(z,\lambda)+\omega f_+(z,\lambda)=C(\lambda)f(z,\lambda).
\end{equation} 
\end{lemma}
\begin{proof} The asymptotics (\ref{asympt}) is available directly or through (\ref{f_pm}) for all three terms within the sector $|\arg{z}|\leq \pi/5-\varepsilon$, $(\varepsilon>0)$. There $(\omega^{\mp1}z)^{5/2}=\omega^{\mp5/2}z^{5/2}$ with $\omega^{-5/2}=\omega^{5/2}=-1$, whence 
$$f_\pm(z,\lambda)=\omega^{\pm 3/4}z^{-3/4}\e^{+(2/5)z^{5/2}}(1+O(z^{-1/2})).$$
Because of $\omega\cdot\omega^{3/4}=-\omega^{-3/4}$, the leading terms cancel in the asymptotics of the l.h.s. in (\ref{sc}), leaving it equal to $z^{-3/4}\e^{(2/5)z^{5/2}}O(z^{-1/2})=o(f_\pm(z))$.
Being moreover a solution of (\ref{DE_turned}), that l.h.s. is a linear combination of $f$ and $f_+$. However the weight of the latter must vanish, because $f(z)=o(f_+(z))$, too.
\end{proof}
Solutions of (\ref{DE}) can be obtained from $f$: The function $f(\alpha z, -\alpha^{-2}\lambda)$ with $z=x$ satisfies that ODE provided $\alpha^5=\im=(-\alpha^{-1})^5$. The choice $\alpha=\im$, $-\alpha^{-2}=1$ corresponds to (\ref{f_psi_relation}), but the solution diverges for $x\to \pm \infty$, since $\im x\to \infty$ in $S_\pm$. Still, that shows that $\lambda$ can be an at most simple eigenvalue of $H$, geometrically speaking. Henceforth we take $\alpha=\e^{\im\pi/10}$ instead. Then the solutions 
\begin{equation}
\label{Psi_pm}
\psi_\pm(z,\lambda):=f(\pm \alpha^{\pm 1}z, -\alpha^{\mp 2}\lambda)
\end{equation}
satisfy the boundary condition: $\psi_\pm(x,\lambda)$ is $L^2$ near $x=\pm\infty$, because $|\pm\pi/10|<\pi/5$. These solutions of (\ref{DE}) can also be written as 
$$\psi_\pm(z,\lambda)=f_\pm(\pm \omega^{\pm 1} \alpha^{\pm 1}z, - \omega^{\pm 3} \alpha^{\mp 2}\lambda)=f_\pm(\im z,\lambda)$$ because $\omega\alpha=\im$, $\omega^3\alpha^{-2}=-1$. 
We remark that (\ref{barf+_f_-}) is restated as 
\begin{equation}
\label{bar_psi_+psi_-}
\overline{\psi_+(-\bar{z},\bar{\lambda})}=\psi_-(z,\lambda),
\end{equation}
reflecting the $PT$-symmetry of Eq.~(\ref{DE}).

\begin{lemma}
\label{Prop_psi_pm}
Let $\lambda\in \mathbb{C}$. 
\begin{enumerate}[i)]
\item Then $\lambda$ is an eigenvalue of $H$ iff $\psi_-(z,\lambda)=\gamma \psi_+(z,\lambda)$, $(z\in \mathbb{C})$ for some $\gamma \in \mathbb{C}$.
\item If so, $ \gamma=-\omega$ and $C(\lambda)=0$ in Eq.~(\ref{sc}).
\item If so, $\psi_-(0,\lambda)=-\omega \psi_+(0,\lambda)$. Conversely, that condition is sufficient for $\lambda$ to be an eigenvalue, provided $|\arg{\lambda}|\leq \pi-\delta$, ($\delta>0$) and $|\lambda|$ is large enough.
\end{enumerate}
\end{lemma}
The first part of (iii) is obvious, given (i, ii). The second part should be contrasted with the fact that, as a rule, it also takes a matching of the first derivatives to make a sufficient condition. 

\begin{proof} i) If $\lambda$ is an eigenvalue then it is geometrically simple, as remarked, and the normalizable solution is unique up to multiples. One implication is thus proved and its converse is obvious. 

ii) Condition (i) is restated as
$$f_-(z,\lambda)=\gamma f_+(z,\lambda)$$ 
after replacing $\im z$ by $z$. Thus (\ref {sc}) yields $(\gamma+\omega)f_+=C(\lambda)f$. By the linear independence of $f$ and $f_+$, the coefficients vanish. 
 
iii) The matching condition and (\ref {sc}) imply $C(\lambda)f(0,\lambda)=0$. However $f(0,\lambda)\neq 0$ for the stated $\lambda$ by (\ref{asymp_f(0,lambda)}).\end{proof}
\begin{lemma}
Let $\psi_+(0,\lambda)=0$. Then $\Re\lambda,\, \Im\lambda >0$.
\end{lemma}
\begin{proof}
Let $H_+=p^2+\im x^3$ on $L^2(\mathbb{R_+})$ with $\mathbb{R}_+=(0,\infty)$ be defined like (\ref{hamiltonian}), but with $p^2$ having Dirichlet boundary conditions at $x=0$. The assumption implies that $\psi_+(x,\lambda)$ is an eigenfunction of $H_+$. Thus 
\begin{equation}
\lambda(\psi_+,\psi_+)=(\psi_+,H\psi_+)=(\psi_+,p^2\psi_+)+\im (\psi_+,x^3\psi_+)\label{psi_skpd}
\end{equation}
has positive real and imaginary parts.
\end{proof} 
As a result the function 
\begin{equation}
\label{h_lambda}
h(\lambda)=\frac{\psi_+(0,\lambda)}{\overline{\psi_+(0,\bar{\lambda})}}
=\frac{\psi_+(0,\lambda)}{\psi_-(0,\lambda)}
\end{equation}
is well-defined and analytic in $\Im\lambda>0$, and continuous up to the boundary $\Im \lambda=0$, where $|h(\lambda)|=1$. For upcoming use we retain the further expression 
\begin{equation}
\label{h_lambda_bis}
h(\lambda)=\frac{f(0,-\alpha^{-2}\lambda)}{f(0,-\alpha^{2}\lambda)}.
\end{equation}
\begin{proposition}
\label{Prop_5}
\begin{equation}
\label{h_smaller1}
|h(\lambda)|<1
\end{equation}
for all $\lambda$ with $\Im\lambda>0$.
\end{proposition}
\begin{proof}[Proof of Theorem~\ref{main_thm}, part (d).] To be shown is $\Re\lambda>0$, $\Im\lambda=0$ for any eigenvalue $\lambda$ of $H$ with eigenfunction $\psi_+ \in L^2(\mathbb{R})$. The first conclusion is as in (\ref{psi_skpd}), since $(\psi_+,p^2\psi_+)>0$ still holds on $L^2(\mathbb{R})$. As for the second, $\bar{\lambda}$ is an eigenvalue too by $PT$-symmetry (\ref{bar_psi_+psi_-}). It thus suffices to rule out $\Im\lambda>0$. Since $h(\lambda)=-\omega^{-1}$ by Lemma~\ref{Prop_psi_pm}, and hence $|h(\lambda)|=1$, we get a contradiction to (\ref{h_smaller1}). 
\end{proof}
The proof of Prop.~\ref{Prop_5} rests on the Phragm\'en-Lindel\"of Theorem, which we state for convenience, see e.g. (\cite{conway}, Cor.~VI.4.2).
\begin{theorem}
\label{THM_10}
Let $G$ be an open sector of angle ${\pi}/{a}, \; (a>{1}/{2})$ with apex 0. Let $F$ be analytic in $G$ and continuous up to the boundary. If 
\begin{equation}
\label{ineq_F<MboundaryG}
|F(z)|\leq M, \qquad (z\in \partial G),
\end{equation}
\begin{equation}
\label{ineq_F_C1C2}
|F(z)|\leq C_1\e^{C_2|z|^b}, \qquad (z\in G)
\end{equation}
for some $M,C_1,C_2>0$ and some $b<a$, then 
\begin{equation}
\label{ineqF_M}
|F(z)|< M,\qquad (z\in G)
\end{equation}
unless $F$ is constant. 
\end{theorem}
We add that the inequality (\ref{ineqF_M}) is most often stated in the wide sense $(\leq)$, but without exceptions. From that statement the one given above follows by the Maximum Modulus principle, see (\cite{conway}, Thm.~IV.3.11).
\begin{proof} [Proof of Prop.~\ref{Prop_5}.] We first note that by $|\e^z|\geq \e^{-|z|}$ Eq.~(\ref{asymp_f(0,lambda)}) implies a lower bound 
\begin{equation}
\label{lower_bondf}
|f(0,\lambda)|\geq c\, \e^{-\tilde{K}|\lambda|^{{5}/{6}}}, \qquad (|\arg \lambda|\leq \pi -\delta), 
\end{equation}
for any $\tilde{K}>K$, some small $c>0$ and all large $|\lambda|$. Next we consider Eq.~(\ref{h_lambda_bis}) and observe that for $\Im \lambda>0$ the point $-\alpha^2\lambda$ appearing in the denominator lies in the sector just mentioned provided $\delta\leq \pi/5$. Combining the lower bound (\ref{lower_bondf}) with the upper bound (\ref{upper_bondf}) we so obtain
$$|h(\lambda)|\leq C'\e^{2\tilde{K}|\lambda|^{{5}/{6}}}, \qquad (\Im \lambda>0)$$
first for large $|\lambda|$ and then, by adjusting $C'$, without any restrictions other than $\Im\lambda>0$. We apply Theorem~\ref{THM_10} (with $\lambda$ instead of $z$) to $h$ on that sector, i.e., with $M=1$, $a=1$, $b={5}/{6}$. Since $b<a$ the conclusion follows by (\ref{ineqF_M}); in fact $h$ can not be a constant, e.g. because of
its asymptotics (\ref{asy_h}). 
\end{proof}

\section{Details of the proof}
We supply the proofs not yet given in the same order as that of the statements they refer to. 
\begin{proof}[Proof of Theorem~\ref{main_thm}, part (a).] We prove below that for $\psi\in \mathcal{D}(p^2)\cap \mathcal{D}(x^3)$
\begin{equation}
\label{norm_ineq}
||(p^2+\im x^3)\psi||^2\geq \frac{1}{2} ||p^2\psi||^2+\frac{1}{2}||x^3\psi||^2-C||\psi||^2,
\end{equation}
where $C>0$. That estimate implies that $H$ is closed. In fact, if $\psi_n\to\psi, \; (\psi_n\in \mathcal{D}(H))$, and $H\psi_n=(p^2+\im x^3)\psi_n$ is Cauchy, then so are $p^2\psi_n$ and $x^3\psi_n$. Since $p^2$ and $x^3$ are closed, $\psi\in\mathcal{D}(p^2)\cap \mathcal{D}(x^3)=\mathcal{D}(H)$. 

The inclusion $C_0^\infty(\mathbb{R})\subset\mathcal{D}(p^2)\cap
\mathcal{D}(x^3)$ is dense.
It thus suffices to prove (\ref{norm_ineq}) for $\psi\in C_0^\infty(\mathbb{R})$. There the following computations are admissible:
\begin{gather*}
||(p^2+\im x^3)\psi||^2= ||p^2\psi||^2+||x^3\psi||^2+(\psi,i[p^2,x^3]\psi),\\
\im [p^2,x^3]=2(px^2+x^2p),\\ 
|(\psi,px^2\psi)|\leq \frac{1}{2} ||p\psi||^2+\frac{1}{2}||x^2\psi||^2,\\
p^2\leq \frac{1}{2}(\varepsilon p^4+\varepsilon^{-1}),\qquad
x^4\leq \frac{1}{3}(2\varepsilon x^6+\varepsilon^{-1}),
\end{gather*}
and the conclusion is by taking $\varepsilon>0$ small enough. \end{proof}
\begin{lemma}
\label{resolvent} The open left complex half--plane is contained in the resolvent set of $H$, i.e., 
$\{z\in\mathbb{C} \mid \Re z<0\}\subset \rho(H)$.
\end{lemma}
The lemma, which is contained in \cite{kato2} as a special case, can phrased by saying that $H$ is maximally accretive. The following proof is adapted from \cite{hu}.
\begin{proof}By $\Re(\psi, H\psi)\ge 0$ we have
\begin{equation}
\label{nr}
\|(H-z)\psi\|\ge (-\Re z)\|\psi\|.
\end{equation}
It suffices to show that $(H-z)\mathcal{D}(H)$ is dense in $L^2(\mathbb{R})$ for some $\Re z<0$. In fact, Eq.~(\ref{nr}) then implies $z\in\rho(H)$ with 
$\|(H-z)^{-1}\|\le (-\Re z)^{-1}$ 
and further extends that conclusion to all $\Re z<0$.

Let $F, G\in C_0^\infty(\mathbb R)$ be real with $FG=G$ and $G(0)=1$, and set $F_n(x)=F(x/n)$ ($n=1, 2\ldots$) and likewise for $G_n$. We consider the truncation $H_n=p^2+\im F_nx^3$ in place of $H$. Then (\ref{nr}) holds true since the potential $\im F_nx^3$ is still imaginary; moreover $z\in\rho(H_n)$ for $\Re z<-M_n$ (and hence $\Re z<0$), since it now has a bound $M_n<+\infty$. We then maintain that$$(H-z)G_n(H_n-z)^{-1}\psi \to \psi,\qquad(n\to\infty)$$
for any $\psi\in L^2(\mathbb{R})$, which implies the density statement for $H$. Indeed, the l.h.s. equals 
$$(G_n(H_n-z)+[p^2, G_n])(H_n-z)^{-1}\psi=G_n\psi-(2\im G'_n(x)p+G''_n(x))(H_n-z)^{-1}\psi,$$
where $\|p\varphi\|^2=\Re(\varphi,p^2\varphi)\le \|\varphi\|\|H_n\varphi\|$ and $\|(H_n-z)^{-1}\|\le (-\Re z)^{-1}$. The limit follows by $G_n\psi\to\psi$ and $\|G'_n\|_\infty,\|G''_n\|_\infty\to 0$. \end{proof}

\begin{proof}[Proof of Theorem~\ref{main_thm}, part (b).] We will show that $(H-z)^{-1}$ is compact for some $z\in\rho(H)$, meaning that it maps any bounded subset of $L^2(\mathbb{R})$ into a compact one. By (\cite{kato}, Thm.~III.6.29) $H$ then has just discrete spectrum. To that end we take $z=-a$ with large $a>0$, whence $-a\in\rho(H)$ by Lemma~\ref{resolvent} and 
\begin{equation*}
||(H+a)\psi||^2\geq \frac{1}{2} ||p^2\psi||^2+\frac{1}{2}||x^3\psi||^2+\frac{1}{2}||\psi||^2.
\end{equation*}
Indeed the l.h.s. equals $||H\psi||^2+2a||p\psi||^2+a^2||\psi||^2$,  where the first term is estimated from below by (\ref{norm_ineq}). At this point the required property follows by Rellich's criterion (\cite{RS4}, Thm.~XIII.65), which we recall: Let $F, G$ be functions on $\mathbb{R}$ diverging at $\infty$. Then the subset consisting of those $\psi\in L^2(\mathbb{R})$ with $||\psi||^2$, $||F(x)\psi||^2$, $||G(p)\psi||^2$ all $\le 1$ is compact. It is applied here with $F(x)=x^3$, $G(p)=p^2$.\end{proof}

\begin{proof}[Proof of Theorem~\ref{main_thm}, part (c).] 
The eigenvalues are geometrically simple, as remarked in Section~\ref{ms}. Let $\lambda_0$ be an eigenvalue of $H$, its eigenfunction being $\psi_-(x,\lambda_0)=-\omega \psi_+(x,\lambda_0)$. Its algebraic multiplicity $m$ is finite by part (b). We need to rule out $m>1$.

We set $\cdot=\partial/\partial\lambda$, differentiate (\ref{sc}) at $\lambda=\lambda_0$, obtain $\dot f_-(z,\lambda_0)+\omega \dot f_+(z,\lambda_0)=\dot C(\lambda_0)f(z,\lambda_0)$ because of Lemma~\ref{Prop_psi_pm}(ii), and hence
\begin{equation}\label{c1}
\dot \psi_-(z,\lambda_0)+\omega \dot \psi_+(z,\lambda_0)=\dot C(\lambda_0)f(\im z,\lambda_0).
\end{equation}
As mentioned earlier, $f(\im x,\lambda)$ is divergent at $x\to\pm\infty$.

We next claim and prove that $m>1$ implies $\dot C(\lambda_0)=0$.
There would exists a (Jordan) generalized eigenvector $\varphi\in \mathcal{D}(H)$ such that $(H-\lambda_0)\varphi=\psi_-$, i.e., 
\begin{equation*}
-\varphi''(x)+(\im x^3-\lambda_0)\varphi(x)=\psi_-(x,\lambda_0).
\end{equation*}
On the other hand, differentiating the ODE (\ref{DE}) with the solution $\psi_\pm(x,\lambda)$ inserted yields
\begin{equation*}
-\dot\psi_\pm''(x,\lambda)+(\im x^3-\lambda)\dot\psi_\pm(x,\lambda)=\psi_\pm(x,\lambda),
\end{equation*}
where $\dot\psi_\pm(x,\lambda)$ is $L^2$ near the same one end $x\to \pm\infty$ as $\psi_\pm(x,\lambda)$, cf.~(\ref{Psi_pm}). That decay indeed follows by the Cauchy estimate (\cite{conway}, IV.2.14) from the uniformity of the bound (\ref{asympt}) locally in $\lambda$. By difference, $\dot\psi_-(\cdot,\lambda_0)-\varphi$ is a solution of the ODE that is $L^2$ near $-\infty$ and hence a multiple of $\psi_-(\cdot,\lambda_0)=-\omega \psi_+(\cdot,\lambda_0)$. Thus $\dot\psi_-(x,\lambda_0)$ is $L^2$ at $x\to +\infty$ as well, and so is the l.h.s. of (\ref{c1}) with $z=x$. We conclude that $\dot C(\lambda_0)=0$.

Differentiating (\ref{h_lambda}) yields
\begin{equation*}
\dot h(\lambda)=
\psi_-(\lambda)^{-1}\Bigl(\dot\psi_+(\lambda)-\frac{\psi_+(\lambda)}{\psi_-(\lambda)}\dot\psi_-(\lambda)\Bigr)=
-\frac{h(\lambda)}{\psi_-(\lambda)}\bigl(\dot\psi_-(\lambda)-h(\lambda)^{-1}\dot\psi_+(\lambda)\bigr)
\end{equation*}
with $\psi_\pm(\lambda)=\psi_\pm(0,\lambda)$, and in particular $\dot h(\lambda_0)=0$ from $h(\lambda_0)^{-1}=-\omega$ and (\ref{c1}). Thus 
\begin{equation*}
-\omega h(\lambda)=1+\beta(\lambda-\lambda_0)^n+O((\lambda-\lambda_0)^{n+1}),\qquad (\lambda\to\lambda_0)
\end{equation*}
for some $n\ge 2$ and $\beta\neq 0$. That limit can be attained within the half--plane $\Im \lambda>0$ in such a way that $\Re\beta(\lambda-\lambda_0)^n>0$. We end up with a contradiction to (\ref{h_smaller1}).
\end{proof}

\begin{proof}[Proof of Theorem~\ref{main_thm}, part (e).] 
For the first part, it suffices to show that the equation $h(\lambda)=-\omega^{-1}$ has infinitely many solutions within the region indicated in Lemma~\ref{Prop_psi_pm}~(iii). In view of part (d) we will look for large $\lambda>0$ as candidates. The points $-\alpha^{\pm 2}\lambda$ seen in Eq.~(\ref{h_lambda_bis}) then lie in the sector where the asymptotics (\ref{asymp_f(0,lambda)}) applies, provided $\delta<\pi/5$. The result is
\begin{equation}\label{asy_h}
h(\lambda)=\e^{-2\pi\im/5}\e^{2\im K\lambda^{5/6}\sin{(2\pi/3)}}(1+O(\lambda^{-5/6}))
\end{equation}
for $|\arg\lambda|<\pi/5-\delta$, besides of $|h(\lambda)|=1$ for real $\lambda$. Hence infinitely many solutions exist. To locate them asymptotically, let $h_0(\lambda)$, ($\lambda\in\mathbb{C}\setminus(-\infty,0]$) be defined by the last expression with $O(\lambda^{-5/6})$ dropped, whence
\begin{equation}\label{Delta}
|(h(\lambda)+\omega^{-1})-(h_0(\lambda)+\omega^{-1})|\le C|h_0(\lambda)||\lambda|^{-5/6}
\end{equation} 
with $C>0$. For comparison, the equation $h_0(\lambda)=-\omega^{-1}$
has the solutions (\ref{l0n}) with $n=0,1, \ldots$. 

Let $1/6<\alpha<5/6$. The set $\{\lambda\mid|h_0(\lambda)+\omega^{-1}|<|\lambda|^{-\alpha}\}$ consists of subsets $G_n\ni\lambda^0_n$ of diameter $O(|\lambda^0_n|^{1/6-\alpha})=o(1)$. 
For $\lambda\in\partial G_n$ with $n$ large we have $|h_0(\lambda)|\le 2$ and (\ref{Delta}) is further bounded by
\begin{equation*}
2C|\lambda|^{-5/6}<|\lambda|^{-\alpha}=|h_0(\lambda)+\omega^{-1}|.
\end{equation*}
Hence, by Rouch\'e's Theorem (\cite{conway}, Thm.~V.3.8), also the equation $h(\lambda)=-\omega^{-1}$ has precisely one solution in each $G_n$ for large $n$. Moreover there are no further solutions $\lambda$. In fact, for large, real $\lambda\notin\cup_{n=0}^\infty G_n$
\begin{equation*}
|h(\lambda)+\omega^{-1}|\ge |h_0(\lambda)+\omega^{-1}|-C|\lambda|^{-5/6}\ge |\lambda|^{-\alpha}-C|\lambda|^{-5/6}>0.
\end{equation*}
\end{proof}

The proof of properties (i-iv) of Prop.~\ref{Properties_of_sol_f} will be a straightforward application of the Liouville-Green approximation as described in Thm.~6.11.1 of \cite{olver}. We reproduce the result for ease of reference.
\begin{theorem}
\label{THM_6} Let $D\subset D_0\subset\mathbb{C}$ be simply connected domains.
Consider the differential equation 
\begin{equation}
\label{f_q}
f''=(q(z)+\tilde{q}(z))f, \qquad ('=\frac{d}{dz})
\end{equation} 
in $z\in D_0$, where $q,\tilde{q}$ are analytic, and suppose that $q$ does not vanish in $D$. Let 
$$\xi'(z)=q(z)^{1/2},$$
and $H(a)\subset D$ be a subset such that for every $z\in H(a)$ there is a finite chain $\mathcal{P}_{a,z}$ in $D$ of regular $C_2-$arcs from $a$ to $z$ along which $\Re\xi(.)$ is non-increasing ($\xi-$progressive path).

Then (\ref{f_q}) admits a solution $f(z)$ analytic in $D_0$ which for $z\in H(a)$ is of the form 
\begin{equation}
\label{f_in_D}
f(z)=q(z)^{-1/4}\e^{-\xi(z)}(1+\varepsilon(z)),
\end{equation}
where $\varepsilon(z)$ satisfies
\begin{equation}
\label{error_control}
|\varepsilon(z)|\leq \e^{V_{a,z}(F)}-1;
\end{equation}
here $F$ is the error--control function with derivative 
$$F'=q^{-1/4}(q^{-1/4})''-\tilde{q}q^{-1/2}$$
and $V_{a,z}(F)$ is the variation of $F$ along $\mathcal{P}_{a,z}$, i.e., 
\begin{equation}
\label{V_a_z}
 V_{a,z}(F)=\int_0^1|F'(z(t))||z'(t)|dt,
\end{equation}
where $[0,1]\ni t\mapsto z(t)$ is a parametrization of $\mathcal{P}_{a,z}$.

The point $a$ may be at infinity, provided the paths $\mathcal{P}_{a,z}$, $(z\in H(a))$ coincide near $a=\infty$ and $V_{a,z}(F)<\infty$.
\end{theorem}

In \cite{olver} the result is stated for $D_0=D$; the formulation given here follows trivially in conjunction with Thm.~5.3.1 there. It will be useful to add two remarks about Theorem~\ref{THM_6} and its proof in \cite{olver}.
\begin{remark}
\label{Rem7}
In the proof of Theorem~\ref{THM_6} the solution just described is singled out at first as the fixed point $f$ of a contraction map, see (\cite{olver}, Eq.~(11.03)), symbolically
\begin{equation*}
 f=K(f).
\end{equation*}
As such it is unique. A property of $f$ which is then implied by that first one may or may not grant uniqueness by itself; but if it does, then $f$ clearly agrees with the solution of the fixed point problem. 
\end{remark}

\begin{remark}
\label{Rem8}
Suppose that $q,\tilde{q}$ depend analytically on a further parameter $\lambda\in D'\subset \mathbb{C}$ and that $z\in H(a,\lambda)$ for some $z\in D$ and (locally) all $\lambda \in D'$. Then the (unique) solution $\fz$ determined by the fixed point problem just mentioned is analytic in $\lambda\in D'$, since the map $K=K(\lambda)$ is. 
\end{remark}
\begin{proof}[Proof of Proposition 2.] We apply Theorem~\ref{THM_6} with
$D_0=\mathbb{C}$ and $D=\mathbb{C}\setminus(-\infty,0]$ by writing Eq.~(\ref{DE_turned}) as (\ref{f_q}) with $q+\tilde{q}$ decomposed as 
\begin{equation}
\label{dec}
 q(z)=z^3, \qquad\tilde{q}(z)=\lambda,
\end{equation} 
and $\xi(z)=(2/5)z^{{5}/{2}}$. The sector $S$: $|\arg{z}|<3\pi/5$ in Eq.~(\ref{asympt}) is mapped bijectively by $z\mapsto \xi(z) $ onto the Riemann surface $\mathcal{R}: \; |\arg{\xi}|<3\pi/2$ described in terms of three sheets in Fig.~\ref{fig_Riemannsheets}.

 \begin{figure}[h]
\centering
 \begin{tikzpicture}[scale=1.2,cap=round,>=latex]

\draw[pattern=north west lines, pattern color=gray, draw=none] (0,1.2) rectangle (-1.2,-1.2);

 \draw[ color=black, thick] (0,0) --(0,-1.2);  
 \draw[loosely dashed, color=black, thick] (0,0) --(0,1.2);      
             \draw (0,-1.4) node {$\mathcal{R}_{-1}$}; 
 \draw[fill=white] (0,0) circle [radius=0.05] node[left] {$0$}; 
 \draw[->,thick, color=black] (-0.4,0.4) -- (0.4,0.4) node[right] {$\text{to }\mathcal{R}_0$};

   
\begin{scope}[shift={(3.5,0)}]

\draw[pattern=north west lines, pattern color=gray, draw=none] (0,1.2) rectangle (1.2,-1.2);
 
 \draw[loosely dashed, color=black, thick] (0,-1.2) --(0,1.2);      
             \draw (0,-1.4) node {$\mathcal{R}_{0}$}; 
 \draw[fill=white] (0,0) circle [radius=0.05] node[left] {$0$}; 
 \draw[->,thick, color=black] (0.4,0.4) -- (-0.4,0.4) node[left] {$\text{to }\mathcal{R}_{-1}$};
  \draw[->,thick, color=black] (0.4,-0.4) -- (-0.4,-0.4) node[left] {$\text{to }\mathcal{R}_{1}$};
    \end{scope}

\begin{scope}[shift={(7,0)}]
\draw[pattern=north west lines, pattern color=gray, draw=none] (0,1.2) rectangle (-1.2,-1.2);

 \draw[ color=black, thick] (0,0) --(0,1.2);  
 \draw[loosely dashed, color=black, thick] (0,0) --(0,-1.2);      
         \draw (0,-1.4) node {$\mathcal{R}_{1}$}; 
 \draw[fill=white] (0,0) circle [radius=0.05] node[left] {$0$}; 
 \draw[->,thick, color=black] (-0.4,-0.4) -- (0.4,-0.4) node[right] {$\text{to }\mathcal{R}_0$};

     \end{scope}

   \end{tikzpicture}
  \caption{Riemann surface $\mathcal{R}$ in the variable $\xi$.}

\label{fig_Riemannsheets}
\end{figure}
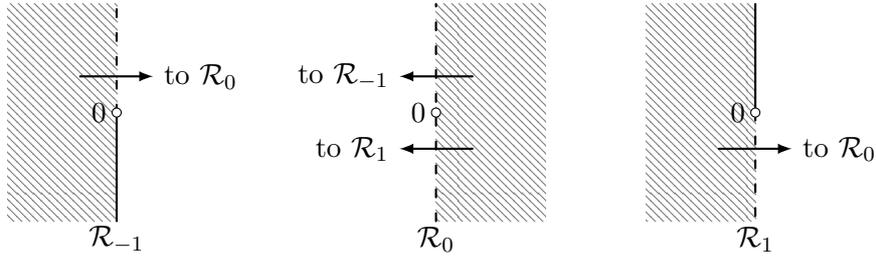

On $\mathcal{R}$ a family of paths is constructed as follows. The point at real infinity, $\xi=+\infty$, on $\mathcal{R}_0$ is joined to $\xi=R>0$ along the real axis, then continued as an arc in the positive (or negative) sense of the circle $|\xi|=R$ till $\xi=-R$ on the sheet $\mathcal{R}_1$ (or $\mathcal{R}_{-1}$), and finally as a line till $\xi=-R\pm \im\infty$, see Fig.~\ref{fig_path}.

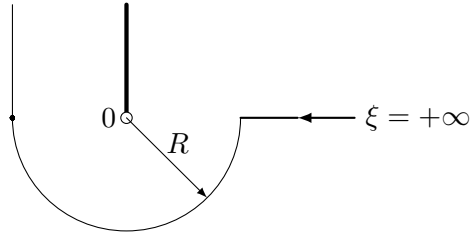
\begin{figure}[h]
\centering
 \begin{tikzpicture}[scale=1.5,cap=round,>=latex]
    \draw[ultra thick,-] (0,0) -- (0,1);
 \draw[-, color=black] (-1,0)--(-1,1) ;
  \draw[<-, color=black, thick] (1.5,0) --(2,0) node[right] {$\xi=+\infty$}; 
  \draw[-, color=black, thick] (1.5,0) --(1,0);          
 \draw [-] (1,0) arc (0:-180:1);
 \draw[fill=white] (0,0) circle [radius=0.05] node[left] {$0$};  
  
    \draw [->] (0,0) -- (-45:1)  ;
    
    \draw (-26:0.5) node {$R$};
    
     \draw plot [only marks, mark size=0.5, mark=*] (-1,0) ;

   \end{tikzpicture}
  \caption{Progressive path in the $\xi$-plane.}

\label{fig_path}
\end{figure} 

Any point $\xi \in \mathcal{R}$ is reached from $\xi=+\infty$ along a unique path, $\mathcal{P}_{+\infty,\xi}$, which is a portion of one of the above paths. Indeed, its parameter $R$ is 
\begin{equation}
\label{param_R}
R= \left\{ 
 \begin{array}{l l}
   |\xi|, & \quad (|\arg{\xi}|\leq \pi),\\
  |\Re\xi|, & \quad (\pi \leq|\arg{\xi}|\leq {3}\pi/{2}).
 \end{array} \right.
\end{equation}
That path is progressive in the sense that $\Re\xi$ is non-increasing along it. By inverting $z\mapsto \xi$ any point $z\in S$ is reached by a $\xi$-progressive path $\mathcal{P}_{+\infty,z}$. Eq.~(\ref{V_a_z}) takes the form
$$ V_{+\infty,z}(F)=\int_{\mathcal{P}_{+\infty,z}}|F'(z)||dz|=\int_{\mathcal{P}_{+\infty,\xi}}\Bigl|\frac{F'(z)}{\xi'(z)} \big|_{z=z(\xi)}\Bigr||d\xi|$$ and a straightforward computation yields 
\begin{gather*}
F'(z)=\frac{21}{16} z^{-{7}/{2}}-\lambda z^{-{3}/{2}},\\
\frac{F'(z)}{\xi'(z)} \bigg|_{z=z(\xi)}=\frac{21}{100} \xi^{-2} -\lambda \bigl(\frac{2}{5}\bigr)^{{5}/{6}}\xi^{-{6}/{5}}.
\end{gather*}
The variation along the whole path bounds that of any portion, $V_{+\infty,z}(F)$. Hence
$$V_{+\infty,z}(F)=O(R^{-1})+\lambda O(R^{-1/5}), \qquad (R \to \infty).$$

We next consider points $z$ in the sector $|\arg{z}|\leq 3\pi/5-\varepsilon, \; (\varepsilon>0)$, as in (\ref{asympt}). Then $|z|\to \infty$ implies that $R$ diverges. In fact $|\Re \xi|>|\xi|\sin{({5}\varepsilon/{2})}$ in the second case (\ref{param_R}), whence $R\gtrsim |z|^{5/2}$ and $$V_{+\infty}(F)=O(|z|^{-5/2})+\lambda O(|z|^{-1/2}), \qquad (|z|\to \infty).$$

In summary: Eqs.~(\ref{f_in_D}, \ref{error_control}) yield $\fz$ satisfying (ii) of Proposition \ref{Properties_of_sol_f}. The solutions $\fz$ and $f_+(z,\lambda)$, cf. (\ref{f_pm}), are decaying, respectively divergent in the sector $S_0$. Hence any solution of Eq.~(\ref{DE_turned}) decaying there agrees with $\fz$ up to a multiple. In particular $\fz$ is uniquely determined by (\ref{asympt}), proving (iii) and, by Remark~\ref{Rem8}, also (i).

We yet next turn to (iv) and to a preliminary, weaker form of (v). Instead of (\ref{dec}) we now use the decomposition
$$q(z)=z^3+\lambda_0, \qquad \tilde{q}(z)=\lambda-\lambda_0$$
with $|\arg \lambda_0|\leq \pi -\delta, \; (\delta>0)$. We apply Theorem~\ref{THM_6} to the sector $D=\lbrace z \mid|\arg{z}|<\delta/{3}\rbrace$, where $q$ nowhere vanishes, to $a=\infty$, and to $H(a)=(0,\infty)$. In fact, for $x>0$ the same path $\mathcal{P}_{+\infty,x}$ as before will be used, since it remains progressive because of $\Re q(x)^{{1}/{2}}>0$ and of
\begin{equation}
\label{xi}
\xi(x,\lambda_0)=\xi(0,\lambda_0)+\int_0^x q(t)^{{1}/{2}}dt.
\end{equation}
We compute 
$$F'(x,\lambda_0)=\frac{45}{16}(x^3+\lambda_0)^{-{5}/{2}}x^4-\frac{3}{2}(x^3+\lambda_0)^{-{3}/{2}}x - (\lambda-\lambda_0)(x^3+\lambda_0)^{-{1}/{2}}$$ and separate the first two terms from the third when estimating the variation:
\begin{gather*}
V_{\infty,x}(F)\leq v(x,\lambda_0)+|\lambda-\lambda_0|\tilde{v}(x,\lambda_0),\\
v(x,\lambda_0)\leq \frac{45}{16}\int_x^\infty \frac{t^4dt}{|t^3+\lambda_0|^{{5}/{2}}} +\frac{3}{2}\int_x^\infty \frac{tdt}{|t^3+\lambda_0|^{{3}/{2}}},\\
\tilde{v}(x,\lambda_0)\leq \int_x^\infty \frac{dt}{|t^3+\lambda_0|^{{1}/{2}}}.
\end{gather*}
We observe
\begin{gather*}
v(x,\lambda_0),\tilde{v}(x,\lambda_0)\to 0, \qquad (x\to \infty),\\
v(0,\lambda_0)=O(\lambda_0^{-5/6}), \quad\tilde{v}(0,\lambda_0)=O(\lambda_0^{-1/6}), \qquad (\lambda_0\to \infty).
\end{gather*}
The solution (\ref{f_in_D}) so constructed for $x>0$, 
\begin{gather*}
f(x,\lambda)=(x^3+\lambda_0)^{-{1}/{4}}\e^{-\xi(x,\lambda_0)}(1+\varepsilon(x,\lambda;\lambda_0)),\\
\varepsilon(x,\lambda;\lambda_0)\leq \e^{V_{\infty,x}(F)}-1,
\end{gather*}
decays as $x\to \infty$.
It thus agrees with (\ref{asympt}) up to a multiple, or even precisely provided $\xi(0,\lambda_0)$ is chosen appropriately, as we will do momentarily. Then $f(x,\lambda)$ is independent of $\lambda_0$, which can be chosen in various ways:
\begin{enumerate}[a)]
\item $\lambda_0=0$. This is the choice made first in (\ref{dec}) together with $\xi(0,0)=0$.
\item $\lambda_0$ with $|\arg \lambda_0|\leq \pi - \delta$. Thus $\xi(x,\lambda_0)-\xi(x,0) \to 0, \; (x\to \infty)$ and hence by (\ref{xi}) 
\begin{gather}\label{K2}
-\xi(0,\lambda_0)=\lim\limits_{x\to \infty}\bigl(\int_0^x \sqrt{t^3+\lambda_0}dt-\int_0^x \sqrt{t^3}dt\bigr)=K'\lambda_0^{{5}/{6}},\\
K'=\int_0^\infty (\sqrt{t^3+1}-\sqrt{t^3})dt.\nonumber
\end{gather}
We observe that $K=K'$, cf.~(\ref{K1}). In fact, let us compare the difference of (\ref{K2}) when the (not admissible) point $\lambda_0=-1$ is approached from the upper and the lower half--plane: In the middle expression, the first radicand changes sign on $0\le t\le 1$ and the difference is $\im K/2-(-\im K/2)=\im K$; on the r.h.s. it is $2\im K'\sin(5\pi/6)= \im K'$. 
\suspend{enumerate}
More specific choices of $\lambda_0$ are the following:
\resume{enumerate}[{[a)]}]
\item $\lambda_0=\lambda$ with $\lambda$ in the same sector. This proves (iv).
\item $\lambda_0>0$ fixed. Then 
\begin{gather*}
f(0,\lambda)=\lambda_0^{-{1}/{4}}\e^{K\lambda_0^{{5}/{6}}}(1+\varepsilon(0,\lambda;\lambda_0)),\\
\varepsilon(0,\lambda;\lambda_0)\leq \e^{v(0,\lambda_0)}\e^{|\lambda-\lambda_0|\tilde{v}(0,\lambda_0)}
\end{gather*}
and we conclude that 
\begin{equation}
\label{ineq_C1C2_1}
|f(0,\lambda)|\leq C_1\e^{C_2|\lambda|}, \qquad (\lambda\in \mathbb{C})
\end{equation}
for some $C_1,C_2>0$. We remark that the bound is weaker than the one claimed in Eq.~(\ref{upper_bondf}).
\end{enumerate}
It remains to prove part (v). This will again be done by means of the Phragm\'en-Lindel\"of principle. Let $G=\lbrace\lambda \mid |\arg{\lambda}-\pi|<\delta\rbrace, \;(\delta>0)$. Outside of $G$ the estimate follows from (iv), in fact with $K$ in place of $\tilde{K}$. In $G\cup\partial G$ we consider the analytic function 
$$F(\lambda)=f(0,\lambda)\e^{-\tilde{K}(-\lambda)^{{5}/{6}}},$$
which satisfies
$$|F(\lambda)|=|f(0,\lambda)|\e^{-\tilde{K}|\lambda|^{{5}/{6}}\cos{(({5}/{6})\arg{(-\lambda)})}}$$ 
with $|\arg(-\lambda)|\le\delta$. Then (\ref{ineq_F<MboundaryG}) (with $\lambda$ instead of $z$) holds for $\delta>0$ so small that $\tilde{K}\cos{(5\delta/6)}\geq K;$ while (\ref{ineq_F_C1C2}) does with $b=1$ by (\ref{ineq_C1C2_1}). The conclusion follows by Theorem~\ref{THM_10}.
\end{proof}

\noindent
{\bf Acknowledgments.} We thank A.~Maltsev and I.M.~Sigal for
discussions and a referee for stressing the algebraic simplicity of
the eigenvalues and for pointing out Refs.~\cite{trinh1, trinh2}.

\end{document}